\documentclass[12pt, letterpaper]{article} %
\usepackage[english]{babel}

\usepackage{geometry}
\usepackage{sectsty}
  \allsectionsfont{\centering\mdseries\scshape}

\usepackage[T1]{fontenc}
\usepackage[utf8]{inputenc}
\usepackage{times}%
\usepackage{xcolor}
  \definecolor{navyblue}{rgb}{0.0, 0.0, 0.5}
  \definecolor{dknavy}{rgb}{0.0, 0.0, 0.33}
\usepackage{setspace}
\usepackage[hyphens]{url}
\usepackage{comment}

\usepackage[footnote]{acronym}
\usepackage{float}
\usepackage{graphicx}
\usepackage{placeins}
\usepackage{booktabs}
\usepackage{threeparttablex}
\usepackage[format=hang, margin=2pt, aboveskip=0pt,
            parskip=3pt, compatibility=false]{caption}
\usepackage[labelformat=simple]{subcaption}

\usepackage{comment}
\usepackage{enumerate}

\usepackage{amsmath}
\usepackage{amssymb}
\usepackage{amsthm}
  \theoremstyle{definition}
  \newtheorem{definition}{Definition}[section]
  
  \newtheorem{remarks}[definition]{Remarks}

  \theoremstyle{plain}
  \newtheorem{proposition}[definition]{Proposition}
  \newtheorem{lemma}[definition]{Lemma}
  \newtheorem{theorem}[definition]{Theorem}
  \newtheorem{corollary}[definition]{Corollary}
  \newtheorem*{main_theorem}{Main Theorem}
\allowdisplaybreaks

\newcommand{\bbC}{\mathbb{C}}
\newcommand{\bbN}{\mathbb{N}}
\newcommand{\bbR}{\mathbb{R}}
\newcommand{\bbP}{\mathbb{P}}

\newcommand{\norm}[1]{\left\lVert #1 \right\rVert}
\newcommand{\modulus}[1]{\left\lvert #1 \right\rvert}
\newcommand{\normiii}[1]{{\left\vert\kern-0.25ex\left\vert\kern-0.25ex\left\vert #1 
    \right\vert\kern-0.25ex\right\vert\kern-0.25ex\right\vert}}
\newcommand{\argument}{\,\cdot\,}

\newcommand{\tp}{{\operatorname{T}}}
\DeclareMathOperator{\cexp}{\operatorname{cexp}} %
\DeclareMathOperator{\clog}{\operatorname{cln}} %
\newcommand{\dx}{\mathrm{d}} %
\newcommand{\dxInt}{\;\mathrm{d}} %
\newcommand{\metric}{\operatorname{d}}
\newcommand{\id}{\operatorname{id}}

\usepackage[authoryear,longnamesfirst]{natbib}
\usepackage[titletoc]{appendix}

\usepackage{titletoc}

\usepackage{tikz-cd}
\usetikzlibrary{cd, er, positioning, calc, arrows, shapes}
\usepackage{nomencl}
  \makenomenclature
\usepackage[pdfborder={0 0 0}, colorlinks=true,
            linkcolor=navyblue, allcolors=navyblue]{hyperref}

\usepackage{epigraph}

\begin{document}
\title{Single and Attractive: Uniqueness and Stability of Economic Equilibria under Monotonicity Assumptions\thanks{Contact: \href{mailto:patrizio.bifulco@uni-wuppertal.de}{\protect\nolinkurl{patrizio.bifulco@uni-wuppertal.de}}, 
\href{mailto:glueck@uni-wuppertal.de}{\protect\nolinkurl{glueck@uni-wuppertal.de}}, 
\href{mailto:okrebs@ethz.ch}{\protect\nolinkurl{okrebs@ethz.ch}}, \href{mailto:bohdan.kukharskyy@baruch.cuny.edu}{\protect\nolinkurl{bohdan.kukharskyy@baruch.cuny.edu}}.
\\ We thank Peter Egger, Benjamin Jung, Wilhelm Kohler, Vincent Lohmann, Davide Suverato for helpful comments and suggestions. Kukharskyy thanks the Alfred P. Sloan Foundation for financial support provided by the CUNY JFRASE. All  errors are our own.
  } 
}

\author{$\quad$Patrizio Bifulco
			\\
	      \textit{$\quad$University of Wuppertal} \and
	    Jochen Gl\"uck$\quad$
			\\
\textit{University of Wuppertal$\quad$} \and \hspace{0.9cm}$\quad\:\:$Oliver Krebs \\ \hspace{0.9cm}\textit{$\quad\:\:$ETH Zürich} \and $\quad\:\:$Bohdan Kukharskyy \\ \textit{$\quad\:\:$City University of New York}}

\date{\today}

\maketitle

\begin{center}
\vspace{-1cm}
\emph{Current version can be found \href{https://drive.google.com/file/d/1qoh-_TJ0tUIOLpJcUrBTsaRB4i85aYmD/view?usp=sharing}{\underline{here}}}
\end{center}

\begin{abstract}

\noindent
This paper characterizes equilibrium properties of a broad class of economic models that allow multiple heterogeneous agents to interact in heterogeneous manners across several markets. Our key contribution is a new theorem providing sufficient conditions for uniqueness and stability of equilibria in this class of models. To illustrate the applicability of our theorem, we characterize the general equilibrium properties of two commonly used quantitative trade models. Specifically, our analysis provides a first proof of uniqueness and stability of the equilibrium in multi-country trade models featuring (i) multiple sectors, or (ii) heterogeneity across countries in terms of their labor cost shares. 
These examples also provide a practical toolkit for future research on how our theorem can be applied to establish uniqueness and stability of equilibria in a broad set of economic models. 

\bigskip
\bigskip
\bigskip
\bigskip

\noindent
\textit{JEL codes:} D50, C62, F11, R13
\\
\textit{Keywords:} Uniqueness, stability, general equilibrium, quantitative trade, networks 
\end{abstract}

\thispagestyle{empty}
\setcounter{page}{0}
\clearpage

\epigraph{\itshape Multiple equilibria are not necessarily useless but, from the standpoint of any exact science, the existence of `uniquely determined equilibrium [...]' is, of course of the utmost importance [...]; without any possibility of proving the existence of uniquely determined equilibrium---or at all events, of a small number of possible equilibria---at however high a level of abstraction, a field of phenomena is really a chaos that is not under analytic control.}{\citep[][p.~969]{Schumpeter}}

\section{Introduction\label{intro}}

Every day economists around the world are in search for answers to pressing policy-relevant questions: Should global value chains be repatriated to avoid future supply chain disruptions, as experienced during the Covid-19 pandemic? How do natural disasters, such as Fukushima or Hurricane Katrina, propagate through multinational firm networks and influence our livelihoods? What is the optimal monetary policy in a %
complex world with multiple sectors and production networks? 
The quest to answer these and many other important questions 
has led to the emergence of an entire class of quantitative general equilibrium (GE) models which aim to put numbers on various shock scenarios and public policies 
through simulations.%
    \footnote{This family of models spans a wide range of economic fields, including international trade, spatial economics, macroeconomics and (production) networks; for an overview of these models, see handbook chapters and review articles by \cite{CostinotRodriguezClare2014}, \cite{AntrasChor2021}, \cite{ReddingRossiHansberg2017}, \cite{Redding2021}, and \cite{CarvalhoTS2019}. In this paper, we refer to this large class of frameworks with an overarching term `quantitative GE models'.}
Yet, a key aspect of  economic modeling has arguably taken a back seat in the quantitative GE literature---the question of uniqueness and stability of equilibria.
Simply put, if a simulation of 
a shock in
a quantitative model 
churns out a given number, what ensures that there is only a single, unique  numerical solution? And is it a stable, attractive equilibrium? In this paper, we
tackle these fundamental questions
by providing a tool to characterize the equilibrium properties of a broad class of cutting-edge quantitative GE models in terms of their uniqueness and stability.%
\footnote{
    Throughout the paper, we use the notion \emph{stability} in the following precise mathematical sense, which encompasses two
    properties.
    We say that the equilibria of the equation $x = F(x)$ are \emph{stable} if
    (i) they are \emph{attractive} in the sense that, for every initial value $x$, the iterates $F^n(x)$ converge to an equilibrium as $n \to \infty$, and additionally
    (ii) every equilibrium $x^*$ is \emph{Lyapunov stable}, which means that for every number $\varepsilon > 0$ there exists a number $\delta > 0$ such that, for all $x \in \bbR^N_{++}$ which satisfy $\norm{x - x^*} < \delta$, one has $\norm{F^n(x) - F^n(x^*)} < \varepsilon$ for all integers $n \ge 0$.
    Note that neither the attractivity nor the Lyapunov stability depends on the specific choice of the norm $\norm{\argument}$ on $\bbR^N$. In fact, both properties are purely topological in nature, i.e., they can be rephrased in terms of open subsets of $\bbR^N$.
}

Specifically, our main contribution is a novel mathematical theorem which provides sufficient conditions for up-to-scale  
uniqueness and global attractivity of the set of fixed point solutions to general equation system in the form $x = F(x)$ with $x \in \bbR^{N}_{++}$ and a function $F: \bbR^N_{++} \to \bbR^N_{++}$, in which the partial reaction of $F(x)$ to changes in $x$ is monotonic.\footnote{Throughout the paper, the equilibrium is said to be unique \emph{up-to-scale} if it is unique up to a choice of numéraire.}
This monotonicity assumption, which we define and discuss in detail below, is natural to many economic models and translates %
to the \emph{partial} cross-derivatives of all equilibrium variables having \emph{constant signs} across their full domain. %
To establish uniqueness and stability using our theorem, one needs to verify a set of requirements on these signs of the Jacobian matrix that must be satisfied by a given GE model.
The major advantage of this approach lies in the fact that despite the Jacobian matrix becoming large in
quantitative GE models with a large number of variables, the signs of its entries are often easily determined, even without explicitly deriving all of its values.%
\footnote{Note that we refer to the Jacobian of $F(x)$ and not to that of the excess demand system $Z(x) \equiv F(x)-x$ on which sign conditions such as the gross-substitute property are often imposed in uniqueness proofs.}

The proof of our theorem relies on a combination of several mathematical methods: the (up-to-scale) uniqueness follows from (i) a characterization of the scaling property of $F$ in terms of eigenvalue properties of the elasticity matrices of the problem and (ii) an application of Perron--Frobenius theory of irreducible matrices which allows us to infer that certain eigenvalues of a matrix are simple.
Attractivity follows by further combining (i) and (ii) with (iii) the choice of an appropriate norm -- a so-called \emph{gauge norm} -- on the variable space which is intimately related to the eigenvectors of the elasticity matrices and gives Lipschitz continuity with constant $1$ for the right-hand side of the equation;
(iv) a quotient space argument which allows us to eliminate the scaling equivariance of the equation without explicitly solving for one of the variables and without fixing the value of any specific function of the variables (i.e.\ without picking a numéraire); 
and (v) a topological technique that is designed to extend, under the Lipschitz continuity established in~(iii), local asymptotic stability to global attractivity.

To illustrate the value added of our theorem, we apply it to a broad class of  quantitative GE models which is commonly used to study the effect of various shock scenarios (e.g., a productivity or trade shock) on different economic outcomes (e.g., countries' production or welfare). Two features of this model class have loomed prominently in numerical simulations: (i) differences across countries regarding their labor cost shares, and (ii) the fact that production takes place in different sectors. Yet, the proof of uniqueness and stability of equilibria in such settings has, to the best of our knowledge, been outstanding to date. We close this gap by setting up a generalized quantitative trade framework which nests as special cases these two real-world features and show how our theorem can be readily applied  to characterize equilibrium properties in these cases. Specifically, we first prove (up-to-scale) uniqueness and stability in a multi-country, one-sector model in which labor and intermediates are used in varying proportions across countries. In the second application, we establish the same properties for a multi-country model with only labor in production but multiple sectors with varying trade elasticities across sectors.

\paragraph*{Related literature.} 

Early economic general equilibrium analysis focused on proofs of existence and culminated in the development of the Arrow-Debreu model \citep[see][]{ArrowDebreu1954, McKenzie1959}. The development of similarly general conditions establishing uniqueness, however, turned out to be significantly more difficult.%
  \footnote{See \cite{Balasko2009} for a historical review of the classical Theory of General Economic Equilibrium as well as its modern developments.} 
A common additional requirement to show uniqueness is the well-known gross-substitution property \citep[see][]{Wald1936}. In pure exchange economies, showing that this property holds for the aggregate excess demand function immediately implies that there is at most one general equilibrium solution \citep[see,  e.g., ][Proposition 17.F.3]{Mas-CollelWinstonGreen1995}. 
 
For example, \cite{AlvarezLucas2007} rely on this property to prove uniqueness in a model similar to \cite{EatonKortum2002}. While this approach may be viable in relatively simple settings with few economic interactions, rewriting more complex models of production economies in terms of an exchange economy becomes excessively difficult and oftentimes unfeasible.\footnote{Another strand of the classical literature  has attempted to analyze uniqueness and multiplicity of economic equilibria based on the index theorem \citep{Dierker1972, Varian1975}. However, this line of work has proven difficult to adapt to more complex settings, i.e. the obtained conditions for uniqueness in production economies are extremely restricting \citep[see][]{Kehoe1985}.} Our approach instead avoids such complexities, by providing conditions that can be applied to the ``stacked'' vector of all variables (and the Jacobian of the respective equation system) without the need to reduce the system to a pure exchange economy first.

In the latter sense, our approach is 
 related to the recent contribution by \cite{AllenArkolakisLi2022}, who use a contraction mapping theorem to derive sufficient conditions for uniqueness of an equilibrium in a class of 
 somewhat more specific network models of the form $x_{ih} = \sum_{j=1}^N f_{ijh}\left(x_{j1},...,x_{jH}\right)$ with $x_{ih}\in\mathbb{R}_{++}^{N\times H}$ and $f_{ijh}:\mathbb{R}_{++}^H\rightarrow \mathbb{R}_{++}$. Intuitively, this system characterizes $N$ heterogeneous agents interacting in $H$ different ways. Instead of %
 reducing the system to a single type of interaction ($H=1$) as in \cite{AlvarezLucas2007},
they assume uniform bounds to the strength of interactions, abstracting thereby from important real-world heterogeneity. While this allows the authors to derive groundbreaking sufficient conditions for uniqueness, it also limits the applicability of their theorem. Specifically, the part of their theorem commonly applied to nominal GE systems of quantitative models (Theorem~1(ii)b), constricts all elasticities $\partial \ln f_{ijh} / \partial \ln x_{jk}$ and hence the strengths of interactions to be constant across heterogeneous agents. This condition is easily violated, for example, in international trade models when the effect of wages ($k$) on prices ($h$) differs across countries ($j$) due to different labor shares in production \citep[as in the seminal workhorse model by][]{CaliendoParro2015}. Similarly, multi-country multi-sector models along the lines of \cite{costinot_what_2012}, with trade elasticities varying by sector, are not covered by their theorem. 
Since our approach does not impose any restrictions on the strength of interactions and relies instead on the full Jacobian describing all interactions between all agent pairs,  we are able to show \emph{up-to-scale uniqueness} of equilibria in multi-country trade models with multiple sectors or with variable labor shares. Moreover, our novel theorem allows us to verify the \emph{stability} of the equilibria, a property which is not covered in part (ii)b of Theorem~1 in \cite{AllenArkolakisLi2022}.

\cite{AllenArkolakis2014} provide sufficient conditions for existence, uniqueness, and stability of a spatial economic equilibrium in continuous space, with a \emph{single sector}, carefully considering the role of external economies of scale. More generally, \cite{AllenArkolakisTakahashi2020} develop a `universal gravity' framework that encompasses a range of common one-sector models and show conditions on the aggregate demand and supply elasticities that ensure existence and uniqueness. Again, one can use our theorem to establish uniqueness in an expanded `universal gravity` framework that comprises, for example, variable labor shares or, more generally, demand and supply elasticities varying by country.

Our work is also related to \cite{Kucheryavyyetal2021a}, who characterize equilibrium properties of existence and uniqueness in a generalized version of the \emph{two-region} spatial economics model that nests the class of models in \cite{AllenArkolakis2014} and further allows economies of scale in a \emph{two-sector} setup (manufacturing and agriculture) as in \cite{Krugman1991}. 
In a \emph{multi-sector} model of international trade with sector-level economies of scale but without intermediates, \cite{Kucheryavyyetal2022} prove uniqueness of equilibrium for the special cases of \emph{frictionless trade} and \emph{two countries}, if the scale elasticity is lower than the inverse of the trade elasticity in every sector. 
In our applications, we show that with constant returns to scale (i.e. a scale elasticity of 0) the solution to the framework from \cite{Kucheryavyyetal2022} is up-to-scale unique for an arbitrary number of countries and with arbitrary trade costs, as long as the underlying graph of the trade network remains strongly connected.%

The remainder of this paper is structured as follows. Section \ref{sec:main-result} defines some terminology and provides our novel theorem. Section \ref{sec:applications} shows how to apply our theorem to quantitative trade models. The proof or our main theorem and necessary intermediate steps are presented in the appendix.

\bigskip

\section{Mathematical main result}
\label{sec:main-result}

Let $N \ge 1$ be an integer and $F: \bbR^N_{++} \to \bbR^N_{++}$ be a continuously differentiable function, where we use the notation $\bbR^N_{++} := \{x \in \bbR^N : \; x_j > 0 \text{ for all } j = 1,\dots,N\}$. Consider the fixed point equation 
\begin{align}
	\label{eq:fixed-point-economic}
	x^* = F(x^*).
\end{align}
We study up-to-scale uniqueness of its solutions, 
whether all solutions are Lyapunov stable, and whether, 
for any initial value $x \in \bbR^N_{++}$, the iterates $F^n(x)$ converge to a solution 
of~\eqref{eq:fixed-point-economic} as $n \to \infty$.%
To do so, we assume a number of mathematical properties of the function $F$. 
For the sake of easier reference, we give names to these properties in the following definition.

\begin{definition}
  \label{def:properties-of-rhs}
   Let $u \in \bbR^N$.
  \begin{enumerate}[\upshape (a)]
        
    \item \label{def:connectedness}
      We say that $F$ \emph{connects all variables} if the modulus of its Jacobian matrix, $\modulus{DF(x)}$, is irreducible for each $x \in \bbR^{N}_{++}$.
        
    \item \label{def:self_interaction}
      We say that $F$ \emph{exhibits self-interaction} if for each $x \in \bbR^{N}_{++}$ there exists an index $j \in \{1, \dots, N\}$ such that $\frac{\partial F_j(x)}{\partial x_j} \not= 0$.
        
    \item\label{def:scaling}
      We say that the function $F$ \emph{scales with exponent $u$} if
      \begin{align*}
          F(c^u x) = c^u F(x)
      \end{align*}
      for all $x \in \bbR^{N}_{++}$ and for all $c \in \bbR_{++}  = (0,\infty)$.%
      \footnote{Here we use the notation $c^u y = (c^{u_1} y_1, \dots, c^{u_N} y_N)$ for $y \in \bbR^N$.}
        
    \item \label{def:monotonicity}
      We say that \emph{the monotonicity behavior} of the function $F$ is \emph{consistent with $u$} if the set $\{1, \dots, N\}$ can be partitioned into two disjoint subsets $\zeta_+$ and $\zeta_-$ such that $u_j \ge 0$ for all $j \in \zeta_+$ and $u_j \le 0$ for all $j \in \zeta_-$ and such that the following property holds:
      For all $x \in \bbR^{N}_{++}$ and all indices $j,k \in \{1,\dots,N\}$:
      \begin{align*}
        \frac{\partial F_j(x)}{\partial x_k} & \ge 0 \quad \text{if both } j \text{ and } k \text{ are located in the same of the sets } \zeta_+,\zeta_-  \\ 
        \frac{\partial F_j(x)}{\partial x_k} & \le 0 \quad \text{if } j \text{ and } k \text{ are not located in the same of the sets } \zeta_+, \zeta_-.
      \end{align*}
    \end{enumerate}
\end{definition}

Note that property~(\ref{def:scaling}) in the definition is trivially satisfied if all entries of $u$ are equal to $0$;
so this condition is only non-trivial if at least one entry of $u$ is different from $0$.
We will show later (Corollary~\ref{cor:signs-vs-spectral-radius}) that if at least one entry of $u$ is different from $0$ and properties~(\ref{def:connectedness}), (\ref{def:scaling}), and~(\ref{def:monotonicity}) are satisfied, then it follows automatically that every entry of $u$ is different from $0$.

Also note that the fact that $F$ scales with some non-zero $u \in \bbR^{N}$ is equivalent to the existence of some $v \in \bbR^{N}_{++}$, 
at least one of whose components is distinct from $1$, 
such that $F(v^\mu x) = v^\mu F(x)$ for all $x \in \bbR^N_{++}$ and all $\mu \in \bbR$.\footnote{With $v^\mu y = (v_1^{\mu} y_1, \dots, v_N^{\mu} y_N)$ for $y \in \bbR^N$ in this case.}

We discuss the mathematical properties underlying these definitions as well as mathematically equivalent alternative specifications in the appendix. For now we limit ourselves to the following remarks.

\begin{remarks}
    \begin{enumerate}[(a)]
        \item 
        In many economics models, it can be expected that $F$ connects all variables since a change of any variable should have -- at least through several intermediaries -- a non-zero impact on the behavior of the other variables.
        
        \item 
        That $F$ exhibit self-interactions is a rather weak technical assumption implying that at least one variable directly influences itself which can often be expected to be satisfied in concrete models in economics.
       
        \item 
        In many economic models, we expect $F$ to scale with a non-zero exponent $u \in \bbR^N$ since the choice of numéraire should not change the behavior of the model.
        
        \item 
        That the monotonicity behavior of $F$ be consistent with $u$ is our main assumption from a mathematical point of view. 
        While the other assumptions can be naturally expected in many models since they directly reflect a certain economic or network property of the model, 
        the monotonicity assumption will turn out to guarantee from a mathematical point of view that one can control the local behavior of $F$ in a certain way 
        (see Subsection~\ref{subsection:monotonicity-condition} in the appendix for more details). Hence, while the other conditions mainly determine the classes of models that fit the setting of our theorem, the monotonicity condition restricts the applicability of the theorem for mathematical reasons.
    \end{enumerate}
\end{remarks}

Using the terminology introduced above we can now formulate our mathematical main result in the following theorem.

\begin{main_theorem}
    Let $u \in \bbR^N$ with at least one entry different from $0$ and assume that $F$ satisfies the properties~\eqref{def:connectedness}, \eqref{def:scaling} and \eqref{def:monotonicity} from Definition~\eqref{def:properties-of-rhs}.
    Assume moreover that the fixed point equation~\eqref{eq:fixed-point-economic} has a solution $x^* \in \bbR^N_{++}$.
    Then one has:
    \begin{enumerate}
        \item[(i)]
        \label{thm:main_theorem_uniqueness}
        \emph{Up-to-scale uniqueness:}
        The solutions of~\eqref{eq:fixed-point-economic} are precisely the vectors in $\bbR^N_{++}$ given by $c^u x^*$ for some $c \in \bbR_{++}$.%
    \end{enumerate}
    If, in addition, $F$  satisfies property~\eqref{def:self_interaction} from Definition~\ref{def:properties-of-rhs}, then one also has:
    \begin{enumerate}
        \item[(ii)]
        \label{thm:main_theorem_convergence}
        \emph{Lyapunov stability and attractivity:}
        Every solution of~\eqref{eq:fixed-point-economic} is Lyapunov stable,
        and for every $x \in \bbR^N_{++}$ the iterates $F^n(x)$ converge to one of the solutions of~\eqref{eq:fixed-point-economic} as $n \to \infty$.
    \end{enumerate}
\end{main_theorem}

\begin{proof}
    We prove the theorem in Appendix~\ref{sec:proof-of-main-result}.
\end{proof}

\section{Applications to trade models}
\label{sec:applications}

In this section, we develop a %
quantitative trade %
framework based on the seminal paper by \cite{EatonKortum2002} and \cite{CaliendoParro2015}. In subsections \ref{subsec:varying_labor_shares} and \ref{subsec:multisec_trade}, we turn to two special cases that are commonly used in the literature to simulate the effects of trade barriers or productivity changes on countries or regions. First, we consider a multi-country, one-sector model in which labor and intermediates are used in varying proportions across different countries. The second application is a multi-country \emph{and} multi-sector model in which labor is the only factor of production. In both cases, we show that the state-of-the art approaches cannot be used to establish uniqueness of the equilibrium. Yet, our main theorem can be readily used to infer both uniqueness and convergence.

\subsection{A general framework}
\label{subsec:general_trade_model}

\paragraph{Setup and Demand}
Each of $J$ countries, indexed by $i,j$, is endowed with $L_i$ worker-consumers that are perfectly mobile across $S$ sectors, indexed by $s$, and are each inelastically supplying one unit of labor. These consumers have Cobb-Douglas preferences across sector bundles that are constant elasticity of substitution (CES) aggregates of an infinite mass of varieties (normalized to 1) produced in each sector. Consumer welfare $U_i$ is given by

\begin{equation}
U_{i} = \sum_{s=1}^S \left(
    \frac{Q_{is}}{\alpha_{is}}
  \right)^{\alpha_{is}}
  \qquad\textrm{with}\quad
Q_{is} = 
    \left(
      \int_0^1
        \left(
          q_{is} \left( \nu \right)
        \right)^{\frac{\sigma_s-1}{\sigma_s}} 
      \mathrm{d}\nu
    \right)^{\frac{\sigma_s}{\sigma_s-1}}\ , \label{eq_CES}
\end{equation}
where $0 \leq \alpha_{is} \leq 1$ is the expenditure share of country $i$'s consumers on the sector $s$ bundle $Q_{is}$ (with $\sum_{s=1}^S \alpha_{is} = 1$), $\nu$ indexes specific varieties, $q_{is}\left(\nu\right)$ denotes consumption of a variety in country $i$ and $\sigma_s>1$ is the elasticity of substitution between varieties in sector $s$. 

\paragraph{Production}
In each country $i$ and sector $s$, perfectly competitive producers can produce all varieties with constant returns to scale and marginal costs $c_{is}$. Implied pricing at marginal costs gives the mill price $p_{is}\left(\nu\right)$ of variety $\nu$ produced in location $i$ and sector $s$ as
\[
  p_{is}\left(\nu\right) =
  \frac{c_{is}}
  {z_{is}\left(\nu\right)}\ ,
\]
where $z_{is}\left(\nu\right)$ denotes the total factor productivity of the respective variety in country $i$ and sector $s$. These productivities are drawn by each country from a country- and sector-specific Fréchet distribution given by the cumulative distribution function (CDF) 
\[
  G_{is}\left(z_{is}\left(\nu\right) \leq z\right) =
  e^{-A_{is} z^{-\theta_s}} \ ,
\]
where $A_{is} > 0$ controls the average of productivities in country $i$'s sector $s$ and $\theta_s > \sigma_s - 1>0$ their spread.

\paragraph{Trade and Price Indices}
We  assume that varieties can be traded worldwide subject to an `iceberg' type transport cost. This type of transport cost is commonly used in trade models and assumes that $\tau_{ijs}>1$ units have to be shipped from country $i$ and sector $s$ in order for one unit to arrive in country $j$. We permit trade costs to become infinite for some $i\neq j$ pairs (implying 0 trade for the respective country-pair and sector), as long as no completely separate trading blocks emerge, and the trade network thus remains 'connected' in the sense of Definition~\ref{def:properties-of-rhs}(\ref{def:connectedness}).

As varieties are not differentiated by origin, consumers will source each variety from the cheapest source after transport costs and we can follow \cite{EatonKortum2002} to derive the price index $P_{is}$ for the sector $s$ consumption bundles $Q_{is}$ in location $i$ from the Fréchet distribution and utility function as

\begin{equation}
    P_{is} = 
    \Gamma
      \left(
        \frac{\theta_s + 1 - \sigma_s}{\theta_s}     
      \right)^{\frac{1}{1-\sigma_s}}
    \left(
      \sum_{j=1}^J 
        A_{js} 
        \left(
          c_{js}
          \tau_{jis}
        \right)^{-\theta_s}
    \right)^{-\frac{1}{\theta_s}} \ ,
    \label{eq_price_index}
\end{equation}
where $\Gamma\left(\cdot\right)$ denotes the gamma function.

Similarly, the share $\pi_{ijs}$ of country $j$'s expenditure on sector $s$ that falls on varieties produced in country $i$ becomes

\begin{equation}
    \pi_{ijs} =
    \frac{A_i 
        \left(
          c_{is}
          \tau_{ijs}
        \right)^{-\theta_s}}
         {\sum_{k=1}^J A_k 
        \left(
          c_{ks}
          \tau_{kjs}
        \right)^{-\theta_s}}.
        \label{eq_import_shares}
\end{equation}

\paragraph{Goods Market Clearing}
In the general equilibrium of our simple trade model, goods markets must clear in each sector. This implies that worldwide expenditure on varieties produced in country $i$ and sector $s$ must be equal to the respective output value. Denoting country $j$'s expenditure on sector $s$ goods by $E_{js}$ we can write 
\begin{equation}
  R_{is} = \sum_{j=1}^J \pi_{ijs} E_{js} \ .
  \label{eq_market_clearing}
\end{equation}

\paragraph{Factor Market Clearing}
In our simple setup, we assume that production takes place by combining labor and sectoral aggregates from potentially all sectors in a Cobb-Douglas fashion. Hence, we have 
\begin{equation}
   \label{eq_marginal_costs}
    c_{is} = w^{\gamma_{is}}_i \prod_{r=1}^S P_{ir}^{\gamma_{irs}},
\end{equation}
where $w_i$ denotes the wage in country $i$, $0\leq\gamma_{is}\leq1$ the cost share of labor in country $i$ and sector $s$ and $0\leq\gamma_{irs}\leq1$ the cost share of sector $r$ intermediates in sector $s$. In each sector $s$ and country $i$, all cost shares must some to one, such that $\gamma_{is} + \sum_{r=1}^S \gamma_{irs} = 1$. The labor market clearing wage can then simply be derived from the total wage sum of all sectors, as
\begin{equation}
    w_i = \frac{\sum_{s=1}^S \gamma_{is} R_{is}}{L_i}\ .
    \label{eq_labor_clearing}
\end{equation}

\paragraph{Balanced Trade}
Finally, we assume that total consumer expenditure in each country equals their income. Under this assumption the total  expenditure (for both final consumption and intermediate usage) on sector $s$ goods in location $i$ is given by
\begin{equation}
   \label{eq_balanced_trade}
    E_{is} = \alpha_{is} w_i L_i + \sum_{r=1}^S \gamma_{isr} R_{ir} \ .
\end{equation}

\paragraph{Equilibrium}
To write down the equilibrium of our trade model, we define so called multilateral resistance terms
\begin{equation}
\label{eq_multilateral_resistance}
\bbP_{is} \equiv P_{is}^{-\theta_s}\qquad \textrm{and}\qquad \Omega_{is} \equiv R_{is} c_{is}^{\theta_s} \ .
\end{equation}
Using these definitions and import shares \eqref{eq_import_shares}, we can rewrite the sectoral goods market clearing equations \eqref{eq_market_clearing}, price index equation \eqref{eq_price_index} and labor market clearing \eqref{eq_labor_clearing} as

\begin{align}
\label{eq_equilibrium_1}
\Omega_{is}&=\sum_{j=1}^{J}\Gamma\left(\frac{\theta_{s}+1-\sigma_{s}}{\theta_{s}}\right)^{-\frac{\theta_{s}}{1-\sigma_{s}}}A_{i}\tau_{ijs}^{-\theta_{s}}\mathbb{P}_{js}^{-1}E_{js}
\\
\label{eq_equilibrium_2}
\mathbb{P}_{is}&=\sum_{j=1}^{J}\Gamma\left(\frac{\theta_{s}+1-\sigma_{s}}{\theta_{s}}\right)^{-\frac{\theta_{s}}{1-\sigma_{s}}}A_{j}\tau_{jis}^{-\theta_{s}}\Omega_{js}^{-1}R_{js} 
\\
\label{eq_equilibrium_3}
w_i &= \sum_{s=1}^S \frac{\gamma_{is}}{L_i} R_{is}\ .
\end{align}

Together with balanced trade \eqref{eq_balanced_trade}, the unit input bundle costs \eqref{eq_marginal_costs} and the definition of multilateral resistance terms \eqref{eq_multilateral_resistance} these equations represent the general equilibrium of this seminal trade model. In the next two subsections, we discuss how our theorem can be applied to establish uniqueness in two commonly used special cases of this equilibrium setup.

\subsection{One sector with varying labor shares}
\label{subsec:varying_labor_shares}
We first turn to a world in which there is only one sector ($S=1$) producing heterogeneous varieties. Given this simplification, we drop the index $s$ where it is not useful in the following. Moreover, with just one sector all $\alpha_{is} = 1$ and all intermediate cost shares ($\gamma_{irs}$) of production are just one minus the labor share, i.e. $1 - \gamma_i$. Consequently, the balanced trade condition \eqref{eq_balanced_trade} collapses to $E_i = R_i$ for each country $i$.  

Moreover, plugging  unit input bundle costs \eqref{eq_marginal_costs} and wages from the labor market clearing condition \eqref{eq_labor_clearing} into $\Omega_i$ from \eqref{eq_multilateral_resistance}, and using therein the definition of $\bbP_{i}$ from \eqref{eq_multilateral_resistance}, yields after rearranging
\[
  R_i = 
  \left(
    \frac{\gamma_i}{L_i}
  \right)^{-\frac{\theta\gamma_i}{1+\theta\gamma_i}}
  \Omega_{i}^{\frac{1}{1+\theta\gamma_{i}}}
  \mathbb{P}_{i}^{\frac{1-\gamma_{i}}{1+\theta\gamma_{i}}} \ .
\]

Using this result in the remaining two equilibrium equations \eqref{eq_equilibrium_1} and \eqref{eq_equilibrium_2} implies that the general equilibrium of this version of the trade model consists in sets of multilateral resistance terms $\bbP_i$ and $\Omega_i$ that solve the following system of equations:

\begin{align}
\label{eq:one_sec_eq_1}
  \Omega_i &= 
    \sum_{j=1}^J
      \Gamma
      \left(
        \frac{\theta + 1 - \sigma}{\theta}     
      \right)^{-\frac{\theta}{1-\sigma}}
      A_i \tau_{ij}^{-\theta}
      \left(
        \frac{\gamma_j}{L_j}
       \right)^{-\frac{\theta\gamma_j}{1+\theta\gamma_j}}
      \Omega_{j}^{\frac{1}{1+\theta\gamma_{j}}}
  \mathbb{P}_{j}^{\frac{1-\gamma_{j}}{1+\theta\gamma_{j}} - 1} 
\\
\label{eq:one_sec_eq_2}
   \mathbb{P}_i &=
    \sum_{j=1}^J
    \Gamma
      \left(
        \frac{\theta + 1 - \sigma}{\theta}     
      \right)^{-\frac{\theta}{1-\sigma}}
        A_j
        \tau_{ji}^{-\theta}
       \left(
        \frac{\gamma_j}{L_j}
       \right)^{-\frac{\theta\gamma_j}{1+\theta\gamma_j}}
      \Omega_{j}^{\frac{1}{1+\theta\gamma_{j}}-1}
  \mathbb{P}_{j}^{\frac{1-\gamma_{j}}{1+\theta\gamma_{j}}} \ .
\end{align}

\paragraph{Uniqueness and Convergence}

We interpret the vector $x$ of our main theorem as a stacked vector of the equilibrium variables $\Omega_i$ and $\bbP_i$ and the function $F(x)$ as the right-hand side of the general equilibrium system given by \eqref{eq:one_sec_eq_1} and \eqref{eq:one_sec_eq_2}. With a slight abuse of notation we will refer to the respective elements of vector $F\left(x\right)$ by using them as an index, e.g. $F(x)_{\Omega_{is}}$.
By our assumption on trade costs, every location is reachable from all other locations, at least through a chain of intermediary trade partners and hence the property of connectedness (Definition~\ref{def:properties-of-rhs}(\ref{def:connectedness})) is satisfied for function $F(x)$ by assumption.
Moreover, the property of self-interaction (Definition~\ref{def:properties-of-rhs}(\ref{def:self_interaction})) is obviously satisfied as in each equation the left hand side variable also appears on the right-hand side and the combined multiplicative constants on the right-hand side are positive for internal flows, i.e. for barriers $\tau_{ii}$ which are finite by assumption.
Economically, it is also clear that the scaling property~\ref{def:properties-of-rhs}(\ref{def:scaling}) must be satisfied for some vector $u \in \bbR^N$ with at least one element different from 0, as we have yet to pick a numéraire. To see this mathematically, we scale all $\Omega_i$ and $\bbP_i$ by the factors $c^{u_{\Omega_i}}$ and $c^{u_{\bbP_i}}$, where we again refer to the respective elements of vector $u$ by using the element as an index. The scaling property will be satisfied if
\begin{align*}
  F\left(c^u x\right)_{\Omega_{i}}
  \hspace{-0.1cm}
  &=
  \hspace{-0.1cm}
  \sum_{j=1}^{J}\Gamma \hspace{-0.1cm}\left(\frac{\theta+1-\sigma}{\theta}\right)^{
  \hspace{-0.1cm}
  -\frac{\theta}{1-\sigma}}
  \hspace{-0.2cm}
  A_{i}\tau_{ij}^{-\theta}\left(\frac{\gamma_{j}}{L_{j}}\right)^{
  \hspace{-0.1cm}
  -\frac{\theta\gamma_{j}}{1+\theta\gamma_{j}}}
  \hspace{-0.2cm}
\left(c^{u_{\Omega_j}}\Omega_{j}\right)^{\frac{1}{1+\theta\gamma_{j}}}\left(c^{u_{\mathbb{P}_j}}\mathbb{P}_{j}\right)^{\frac{1-\gamma_{j}}{1+\theta\gamma_{j}}-1}
\hspace{-0.1cm}
= c^{u_{\Omega_i}} F\left(x\right)_{\Omega_{i}}
\\
F\left(c^u x\right)_{\mathbb{P}_i}
 \hspace{-0.1cm}
  &=
  \hspace{-0.1cm}
\sum_{j=1}^{J}\Gamma\hspace{-0.1cm}\left(\frac{\theta+1-\sigma}{\theta}\right)^{\hspace{-0.1cm}-\frac{\theta}{1-\sigma}}
\hspace{-0.2cm}
A_{j}\tau_{ji}^{-\theta}\left(\frac{\gamma_{j}}{L_{j}}\right)^{
\hspace{-0.1cm}
-\frac{\theta\gamma_{j}}{1+\theta\gamma_{j}}}
\hspace{-0.2cm}
\left(c^{u_{\Omega_j}}\Omega_{j}\right)^{\frac{1}{1+\theta\gamma_{j}}-1}\left(c^{u_{\mathbb{P}_j}}\mathbb{P}_{j}\right)^{\frac{1-\gamma_{j}}{1+\theta\gamma_{j}}}
\hspace{-0.1cm}
=
c^{u_{\bbP_i}} F\left(x\right)_{\bbP_{i}},
\end{align*}
which can be easily shown to hold if for all $i$ we set $u_{\bbP_i} = k$ and $u_{\Omega_i} = -\frac{\theta}{1+\theta}k$ for any value $k \neq 0$. Finally, this also implies that for any $k \neq 0$ the sign of all $u_{\bbP_i}$ will be the same and it will be different from the sign of all $u_{\Omega_i}$. Hence, the groups $\zeta_+$ and $\zeta_-$ of Definition~\ref{def:properties-of-rhs}(\ref{def:monotonicity}) are respectively formed by all $\Omega_{j}$ and all $P_{j}$. Consequently, the monotonicity behavior will be consistent with $u$ (see Definition~\ref{def:properties-of-rhs}(\ref{def:monotonicity})) if all partial derivatives of \eqref{eq:one_sec_eq_1} are (weakly) positive with respect to any $\Omega_j$ and (weakly) negative with respect to any $\bbP_j$ and vice versa for equation \eqref{eq:one_sec_eq_2}. Given that all constant terms in equations \eqref{eq:one_sec_eq_1} and \eqref{eq:one_sec_eq_2} are positive or 0, it is sufficient to verify that the exponents on $\Omega_j$ and $\bbP_j$ have the respective signs:
\begin{align*}
  \frac{1}{1+\theta\gamma_{j}} & >0\quad\Rightarrow \quad\frac{\partial F\left(x\right)_{\Omega_{i}}}{\partial\Omega_{j}}\geq0 
&
  \frac{1-\gamma_{j}}{1+\theta\gamma_{j}}-1 & <0\quad\Rightarrow \quad\frac{\partial F\left(x\right)_{\Omega_{i}}}{\partial \mathbb{P}_{j}}\leq0
\\
  \frac{1}{1+\theta\gamma_{j}}-1 & <0\quad\Rightarrow \quad\frac{\partial F\left(x\right)_{\mathbb{P}_{i}}}{\partial\Omega_{j}}\leq0 
&
  \frac{1-\gamma_{j}}{1+\theta\gamma_{j}} & >0\quad\Rightarrow \quad\frac{\partial F\left(x\right)_{\mathbb{P}_{i}}}{\partial\mathbb{P}_{j}}\geq0
\end{align*}
Clearly, variables in the same group ($\zeta_+$, $\zeta_-$) influence each other (weakly) positively and variables from different groups each other (weakly) negatively.

Since our equation system \eqref{eq:one_sec_eq_1}, \eqref{eq:one_sec_eq_2} thus satisfies properties (a)-(d) of Definition \ref{def:properties-of-rhs} for some $u\in\bbR^N$ with at least one entry different from $0$, our main theorem implies that any equilibrium solution must be up-to-scale unique and can be obtained by iterating the right-hand side for an initial guess of the solution. 

\paragraph{Previous Literature}
To the best of our knowledge, we are the first to show uniqueness and stability for multi-country Eaton-Kortum type model with varying labor cost shares in production across countries. Closest to our result is the main theorem of \cite{AllenArkolakisLi2022} which can be applied to the special case of our model in which labor cost shares are assumed to be equal across all countries (i.e. $\gamma_i = \gamma\ \forall i$).In this case, one can turn to their Theorem 1 (iib) and prove (column-wise up to scale) uniqueness by showing that the matrix 
\[
  \mathbf{A} = \begin{pmatrix} 
  \frac{1}{1+\theta\gamma} & \left\vert\frac{1-\gamma}{1+\theta\gamma}-1\right\vert \\ 
  \left\vert\frac{1}{1+\theta\gamma}-1\right\vert & \frac{1-\gamma}{1+\theta\gamma}
  \end{pmatrix}
\]
which consists of the absolute values of the respective right-hand side exponents of the equilibrium equation system, has a spectral radius equal to 1, with the latter result following from the Collatz-Wielandt formula \citep[p. 670]{Meyer2000} and the fact that the columns of $\mathbf{A}$ sum to 1. However, when labor shares are instead allowed to be country-specific in accordance with real world data, the dependence of the exponents on $j$ requires one to instead rely on Theorem 1 (iia) in \cite{AllenArkolakisLi2022}. This theorem shows uniqueness if a matrix $\mathbf{A}^{max}$ consisting of upper bounds of the respective absolute exponents in our equilibrium system across all $j$ (and strict upper bound for at least one $j$) has spectral radius of 1. However, this condition is violated for matrix
\[
  \mathbf{A}^{max} = \begin{pmatrix} 
  \max_j \frac{1}{1+\theta\gamma_{j}} + \varepsilon_j & \max_j \left\vert\frac{1-\gamma_{j}}{1+\theta\gamma_{j}}-1\right\vert + \varepsilon_j \\ 
  \max_j\left\vert\frac{1}{1+\theta\gamma_{j}}-1\right\vert + \varepsilon_j & \max_j\frac{1-\gamma_{j}}{1+\theta\gamma_{j}} + \varepsilon_j
  \end{pmatrix}
\]
with some $\varepsilon_j > 0$, since both columns sum to values larger 1 and, by the Collatz-Wielandt formula, the spectral radius is thus larger than 1. Hence, in contrast to our novel theorem, the theorem of \cite{AllenArkolakisLi2022} cannot be applied to show uniqueness in this setup.%
\footnote{
We note that, conversely, in many cases in which Theorem 1 (ii)b of \cite{AllenArkolakisLi2022} is applicable, our theorem can be applied as well. Yet, there are the following two exceptions: 
(1) If property~\eqref{def:connectedness} of our Definition~\ref{def:properties-of-rhs} is not satisfied, then our theorem cannot be applied, while there are no assumptions of this type in \cite{AllenArkolakisLi2022} (as a trade-off, since no irreducibility is assumed in \cite{AllenArkolakisLi2022}, the up-to-scale uniqueness established there is only column wise).
(2) There are cases in which the matrix $\modulus{A}$ can have spectral radius $1$ but $\rho\left(A\right) < 1$, and thus our scaling condition from Definition~\ref{def:properties-of-rhs}(\ref{def:scaling}) will be violated. In these cases, our main theorem cannot be applied either, while Theorem~1(ii)b of \cite{AllenArkolakisLi2022} is still applicable and yields (column-wise) up-to-scale uniqueness. We note that in those cases one would, due to the lack of scaling, intuitively even expect uniqueness rather than up-to-scale uniqueness.}

\subsection{Multisector model}
\label{subsec:multisec_trade}
We next return to our multi-sector trade model developed in Subsection~\ref{subsec:general_trade_model} but assume instead that production takes place without intermediates. This implies that all labor shares $\gamma_{is}$ are equal to 1 and all intermediate shares $\gamma_{irs}$ are equal to 0. Consequently, sectoral expenditures \eqref{eq_balanced_trade} collapse to $E_{is} = \alpha_{is}w_i L_i$. Moreover, unit input bundle costs \eqref{eq_marginal_costs} in all sectors now consist simply of the wage, i.e. $c_{is}=w_i$ and the definition of multilateral resistance simplifies to $\Omega_{is}=R_{is} w_i^{\theta_s}$, yielding  $R_{is} = \Omega_{is} w_i^{-\theta_s}$. Plugging sectoral expenditures and revenues into the remaining equilibrium equations \eqref{eq_equilibrium_1} through \eqref{eq_equilibrium_3} results in

\begin{align}
\label{eq_multsec_equilibrium_1}
    \Omega_{is} &=
    \sum_{j=1}^J 
      \Gamma
      \left(
        \frac{\theta_s + 1 - \sigma_s}{\theta_s}     
      \right)^{-\frac{\theta_s}{1-\sigma_s}}
      A_{is} \tau_{ijs}^{-\theta_s} \alpha_{js} L_{j} 
         \bbP_{js}^{-1} W_j^\frac{1}{1+\Theta}  
  \\
  \label{eq_multsec_equilibrium_2}
  \bbP_{is} &=
    \sum_{j=1}^J
    \Gamma
      \left(
        \frac{\theta_s + 1 - \sigma_s}{\theta_s}     
      \right)^{-\frac{\theta_s}{1-\sigma_s}}
        A_{js}
        \tau_{jis}^{-\theta_s}
        W_j^\frac{-\theta_s}{1+\Theta}
  \\
  \label{eq_multsec_equilibrium_3}
  W_i &= 
    \sum_{r=1}^S 
      L_i^{-1}\Omega_{ir}W_i^\frac{\Theta-\theta_r}{1+\Theta}\ ,
\end{align}
where we transformed the labor market clearing condition by introducing $W_i = w_i^{1+\Theta}$ and defining the constant $\Theta \equiv \sum_{s=1}^S \theta_s$. 

\paragraph{Uniqueness and Convergence}

We again interpret the vector $x$ of our main theorem as a stacked vector of the equilibrium variables $\Omega_{is}$, $\bbP_{is}$ and $W_i$, as well as the function $F(x)$ as the right-hand side of the general equilibrium system given by \eqref{eq_multsec_equilibrium_1} through \eqref{eq_multsec_equilibrium_3}, and use indices to refer to the respective elements of $F\left(x\right)$. 
By our assumption on trade costs, every country-sector is reachable from all other country-sectors, at least through a chain of intermediary trade partners and hence the property of connectedness (Definition~\ref{def:properties-of-rhs}(\ref{def:connectedness})) is satisfied for function $F(x)$ by assumption.
Moreover, the property of self-interaction (Definition~\ref{def:properties-of-rhs}(\ref{def:self_interaction})) is obviously satisfied as in each equation of \eqref{eq_multsec_equilibrium_3} the left hand side variable also appears on the right-hand side.
Economically, it is also clear that the scaling property~\ref{def:properties-of-rhs}(\ref{def:scaling}) must be satisfied for some vector $u \in \bbR^N$ with at least one element different from $0$, as we have yet to pick a numéraire. To see this mathematically, we scale all $\Omega_{is}$, $\bbP_{is}$ and $W_i$ by the factors $c^{u_{\Omega_{is}}}$, $c^{u_{\bbP_{is}}}$ and $c^{u_{W_i}}$, where we again refer to the respective elements of vector $u$ by using the element as an index. The scaling property will be satisfied if
\begin{align*}
F\left(c^u x\right)_{\Omega_{is}}& =\sum_{j=1}^{J}\Gamma\left(\frac{\theta_{s}+1-\sigma_{s}}{\theta_{s}}\right)^{-\frac{\theta_{s}}{1-\sigma_{s}}}A_{i}\tau_{ijs}^{-\theta_{s}}\alpha_{js}L_{j}\left(c^{u_{\mathbb{P}_{js}}}\mathbb{P}_{js}\right)^{-1}\left(c^{u_{W_{j}}}W_{j}\right)^{\frac{1}{1+\Theta}} = c^{u_{\Omega_{is}}}F\left(x\right)_{\Omega_{is}}
\\
F\left(c^u x\right)_{\mathbb{P}_{is}} & =\sum_{j=1}^{J}\Gamma\left(\frac{\theta_{s}+1-\sigma_{s}}{\theta_{s}}\right)^{-\frac{\theta_{s}}{1-\sigma_{s}}}A_{j}\tau_{jis}^{-\theta_{s}}\left(c^{u_{W_{j}}}W_{j}\right)^{\frac{-\theta_{s}}{1+\Theta}} = 
c^{u_{\bbP_{is}}}F\left(x\right)_{\bbP_{is}}
\\
F\left(c^u x\right)_{W_{i}} & =\sum_{r=1}^{S}\frac{1}{L_{i}}c^{u_{\Omega_{ir}}}\Omega_{ir}\left(c^{u_{W_{i}}}W_{i}\right)^{\frac{\Theta-\theta_{r}}{1+\Theta}}= c^{u_{W_{i}}}F\left(x\right)_{W_{i}},
\end{align*}
which can easily shown to hold if for all $i$ and $s$ we set $u_{W_i} = k$,  $u_{\mathbb{P}_{is}}=-\frac{\theta_{s}}{1+\Theta}k$ and $u_{\Omega_{is}}=\frac{1+\theta_{s}}{1+\Theta}k$ for any value $k \neq 0$. Finally, such a scaling vector also implies that for any $k \neq 0$ the sign of $u_{\Omega_{is}}$ is the same as the sign of $u_{W_i}$ and differs from the sign of $u_{\bbP_{is}}$ for all $i$ and $s$. Hence, the groups $\zeta_+$ and $\zeta_-$ of Definition~\ref{def:properties-of-rhs}(\ref{def:monotonicity}) are respectively formed by all $\Omega_{js}$ and $W_j$ on the one hand, and all $P_{js}$ on the other hand. Consequently, the monotonicity behavior will be consistent with $u$ (see Definition~\ref{def:properties-of-rhs}(\ref{def:monotonicity})) if all partial derivatives of \eqref{eq_multsec_equilibrium_1} are (weakly) positive with respect to any $W_j$ and (weakly) negative with respect to any $\bbP_{js}$, all partial derivatives of \eqref{eq_multsec_equilibrium_2} (weakly) negative with respect to any $W_j$ and all partial derivatives of \eqref{eq_multsec_equilibrium_3} (weakly) positive with respect to any $\Omega_{js}$ or $W_j$. Given that all constant terms in equations \eqref{eq:one_sec_eq_1} and \eqref{eq:one_sec_eq_2} are positive or 0, the sign of the exponents of $\Omega_{js}$, $\bbP_{js}$ and $W_j$ on the right-hand side of the equation system determine the sign of the respective partial derivative:
\begin{align*}
 \frac{\partial F\left(x\right)_{\Omega_{is}}}{\partial\Omega_{js}}&=0 &
 \frac{\partial F\left(x\right)_{\Omega_{is}}}{\partial\mathbb{P}_{js}}&\leq0 &  \frac{\partial F\left(x\right)_{\Omega_{is}}}{\partial W_{j}}&\geq0\\
 \frac{\partial F\left(x\right)_{\mathbb{P}_{is}}}{\partial\Omega_{js}}&=0 &
 \frac{\partial F\left(x\right)_{\mathbb{P}_{is}}}{\partial\mathbb{P}_{js}}&=0 &  \frac{\partial F\left(x\right)_{\mathbb{P}_{is}}}{\partial W_{j}}&\leq0\\
\frac{\partial F\left(x\right)_{W_{i}}}{\partial\Omega_{js}}&\geq0 &
\frac{\partial F\left(x\right)_{W_{i}}}{\partial\mathbb{P}_{js}}&=0 & 
\frac{\partial F\left(x\right)_{W_{i}}}{\partial W_{i}}&\geq0,
\end{align*}
whereby the last inequality follows from the fact that $\frac{\Theta-\theta_{r}}{1+\Theta}>0$. Clearly, variables in the same group ($\zeta_+$, $\zeta_-$) influence each other (weakly) positively and variables from different groups each other (weakly) negatively.

Since our equation system \eqref{eq_multsec_equilibrium_1} through \eqref{eq_multsec_equilibrium_3} thus satisfies properties (a)-(d) of Definition~\ref{def:properties-of-rhs} for some $u\in\bbR^N$ with at least one entry different from $0$, our main theorem implies that any equilibrium solution must be up-to-scale unique and can be obtained by iterating the right-hand side for an initial guess of the solution. 

\paragraph{Previous Literature}
We note here that our equation system satisfies the general form given in Theorem (iib) of \cite{AllenArkolakisLi2022}. However, in the current multi-sector setting, exponents of the equilibrium variables on the right-hand side of the system depend on the sector $s$.
Therefore, we cannot apply their Theorem 1 (iib) to establish (column-wise up-to-scale) uniqueness. Moreover, if we proceed as in section \ref{subsec:varying_labor_shares} and form a matrix equivalent to their matrix $\mathbf{A}$ consisting of the respective absolute values of exponents for each sector individually, each of these matrices has a spectral radius of 1. Hence, with heterogeneous sectoral $\theta_s$ taking for each element the largest entry across these matrices (plus some $\varepsilon>0$) as an upper bound, the resulting matrix $\mathbf{A}^{max}$ has a spectral radius larger 1 and part (iia) of their theorem cannot be applied either.

\section{Conclusion}

We have developed a novel mathematical theorem which provides sufficient conditions for uniqueness and stability of the fixed point solution of a 
very general equation system. We have shown that this theorem can be readily applied to establish these properties for workhorse quantitative trade models for which uniqueness and stability could not be proven with existing methods. Our
novel
theorem
thus breaks new ground in terms of applicability and increases our confidence in numerous simulations of trade models used to study a plethora of relevant policies and shocks, which take the numeric results as the unique model outcomes without an established proof.

\newpage
\let\oldbibliography\thebibliography
\renewcommand{\thebibliography}[1]{
	\oldbibliography{#1} \setlength{\itemsep}{0pt}}
\bibliography{main}

\newpage

\begin{appendices}
\section*{Appendix}

The proof of our main theorem requires several building blocks. 
First, in Appendix~\ref{section:perron-frobenius-theory},
we recall a number of mathematical preliminaries, including some aspects of Perron--Frobenius theory for non-negative and irreducible matrices. 
In Appendix~\ref{sec:reformulations}, we then analyze the mathematical properties of $F$, the right-hand side of our fixed point equation, in particular those listed in Definition~\ref{def:properties-of-rhs}.
In Appendix~\ref{section:stability-non-expansive}, we prove a general attractivity result for fixed points of a certain class of mappings on metric spaces.
After these preparations, the proof of the main theorem
 is provided in Appendix~\ref{sec:proof-of-main-result}.

\section{Eigenvalues of non-negative matrices}
\label{section:perron-frobenius-theory}

\subsection{Spectral theoretic terminology}

In our setup, we consider matrices $A \in \bbR^{N \times N}$ with real coefficients. 
Recall that a number $\lambda \in \bbC$ is called an \emph{eigenvalue} of $A$ if there exists a vector $v \neq 0$ 
(called an \emph{eigenvector} of $A$ for the eigenvalue $\lambda$) such that $Av=\lambda v$. Note that both $\lambda$ and $v$ might be complex, even though all entries of $A$ were assumed to be real.
The set of all eigenvalues is denoted with $\sigma(A)$ and is often called the \emph{spectrum} of $A$. 
Moreover, the subspace of all vectors $v \in \bbC^N$ that satisfy $Av=0$ is called the \emph{kernel}, $\ker(A)$, of the matrix $A$.

Clearly a complex number $\lambda \in \bbC$ is an eigenvalue of some matrix $A \in \bbR^{N \times N}$ if and only if $\ker(\lambda-A)$ contains an element distinct from $0$ (here, we use the common notation $\lambda - A := \lambda \id_N - A$, where $\id_N \in \bbR^{N \times N}$ denotes the identity matrix).
Equivalently, one would require that the matrix $\lambda-A$ is not \emph{invertible}. 
In terms of \emph{determinants}, this means that $\lambda$ is an eigenvalue of $A$ if and only if the determinant $\det(\lambda-A)$ is $0$. 

The polynomial function $\lambda \mapsto \det(\lambda-A)$ is called \emph{characteristic polynomial of $A$}. It is often denoted by $\chi_A$, and 
its roots are precisely the eigenvalues of $A$. 
By the fundamental theorem of algebra, one can factorize the characteristic function as
\begin{align*}
    \chi_A(\lambda) = \prod_{i=1}^N (\lambda-\lambda_i),
\end{align*} 
where $\sigma(A) = \{ \lambda_1,\dots,\lambda_N \}$. 
Here some of the numbers $\lambda_1, \dots, \lambda_N$ might coincide. 
The number of occurrences of an eigenvalue in the factorized representation of the characteristic polynomial $\chi_A$ is called the \emph{algebraic multiplicity} of this eigenvalue. 

There is also a second notion of multiplicity which will be particularly important in our setting.
For an eigenvalue $\lambda \in \sigma(A)$ the vector subspace $\ker(\lambda - A)$ of $\bbC^N$ is called the \emph{eigenspace} of $A$ for the eigenvalue $\lambda$. 
It consists precisely of the eigenvectors of $A$ for the eigenvalue $\lambda$ and of the zero vector.
The dimension of the eigenspace is called the \emph{geometric multiplicity} of the eigenvalue $\lambda$. 

An eigenvalue $\lambda$ of $A$ is called \emph{algebraically simple} if its algebraic multiplicity is $1$. It is called \emph{geometrically simple} if its geometric multiplicity is $1$. 
It is a standard fact from linear algebra that the geometric multiplicity of an eigenvalue is always smaller than or equal to the geometric multiplicity. 
Hence, if an eigenvalue is algebraically simple, then it is automatically geometrically simple 
(but the converse implication does not hold, in general).

Lastly, we recall that the so-called \emph{spectral radius} $\rho(A)$ of a matrix $A$ is defined to be the maximal absolute value of all eigenvalues, i.e.,
\[
    \rho(A) := \max_{\lambda \in \sigma(A)} \vert \lambda \vert.
\]
Among all closed disks in the complex plane $\bbC$ with center $0$, the one with radius $\rho(A)$ is the smallest one which contains the spectrum $\sigma(A)$. This explains the terminology \emph{spectral radius}. A property of the spectral radius is, that for any given norm on $\bbC^N$, the induced matrix norm satisfies the inequality
\begin{align*}
    \rho(A) \le \norm{A}
\end{align*}
for every $N \times N$-matrix $A$.
However, there are matrices $A$ for which no norm exists that lets this inequality become an equality.

\subsection{Perron--Frobenius theory}

For certain matrices with non-negative entries there are useful criteria to ensure that the eigenvalues with modulus $\rho(A)$ are algebraically (and thus also geometrically) simple. 
This is part of the so-called \emph{Perron--Frobenius theory}, which we will now discuss in some excerpts.

Let $A,B \in \mathbb{R}^N$. 
We write $B \ge A$ if each entry of $B$ is larger than or equal to the corresponding entry of $A$. 
In particular, the notation $B \ge 0$ means that each entry of $B$ is non-negative;
in this case, $B$ is said to be \emph{non-negative}.
We define $\modulus{A} \in \mathbb{R}^N$ to be the matrix whose entries are the absolute values of the entries of $A$. Similar notation and terminology is also used for vectors in $\bbR^N$.

Now, let $0 \leq B \in \bbR^{N \times N}$. 
One can associate a directed graph with $N$ vertices to $B$, 
where there is an edge from the $k$-th vertex to the $j$-th vertex if and only if the entry $B_{jk}$ (i.e., the entry of $B$ in row $j$ and column $k$) is non-zero. 
The matrix $B$ is called \emph{irreducible} if the associated graph is \emph{strongly connected}, which means that one can walk from any vertex to any other vertex along the edges of the graph (where one is not allowed to walk oppositely to the direction of an edge).

The following two results for non-negative and for irreducible matrices are part of what is often referred to as the \emph{Perron--Frobenius theorem}:

\begin{theorem}[Perron--Frobenius for non-negative matrices]
    \label{thm:perron-frobenius-nonnegative}
    Let $0 \leq B \in \bbR^N$.
    Then the spectral radius $\rho(B)$ is an eigenvalue of $B$ and has at least one eigenvector $u \ge 0$.
\end{theorem}

\begin{proof}
    See for instance \cite[Theorem~4.2 on p.\,14]{Minc1988}.
\end{proof}

Recall from Section~\ref{sec:main-result} that we denote the set of vectors in $\bbR^N$ whose entries are all positive by $\bbR^N_{++}$.

\begin{theorem}[Perron--Frobenius for irreducible matrices]
    \label{thm:perron-frobenius-irreducible}
    Let $0 \leq B \in \bbR^N$ be irreducible. 
    Then:
    \begin{enumerate}[\upshape (a)]
        \item 
        All those eigenvalues of $B$ that have modulus $\rho(B)$ are algebraically (and thus geometrically) simple.
        
        \item
        There exists a vector $v \in \bbR^N_{++}$ which spans the eigenspace $\ker(\rho(B) - B)$, 
        
        \item 
        The matrix $B$ does not have a non-negative eigenvector for any eigenvalue except $\rho(B)$.
    \end{enumerate}
\end{theorem}

\begin{proof}
    (a)
    See for instance \cite[Theorem~1.6.5(ii) on p.\,22]{Schaefer1974}.
    
    (b)
    The existence of an eigenvector $v \in \bbR^N_{++}$ for the eigenvalue $\rho(B)$ can, for instance, be found in \cite[Theorem~4.1 on p.\,11]{Minc1988}.
    The fact that $v$ spans $\ker(\rho(B)-B)$ follows from the geometric simplicity of the eigenvalue $\rho(B)$ stated in~(a).
    
    (c) 
    This follows, for instance, from \cite[Theorem~4.4]{Minc1988}.
\end{proof}

An irreducible matrix $0 \leq B \in \bbR^{N \times N}$ is called \emph{primitive} if $\rho(B)$ is the only eigenvalue of $B$ with modulus $\rho(B)$. 
The following criterion for primitivity will be useful in our proof.

\begin{proposition}
    \label{prop:primitivity}
    Let $0 \leq B \in \bbR^{N \times N}$ be irreducible. 
    If at least one diagonal entry of $B$ is non-zero, then $B$ is primitive.
\end{proposition}
\begin{proof}
    See for instance \cite[Theorem~1.6.5, Corollary~2 on p.\,23]{Schaefer1974}.
\end{proof}

\section{Analysis of the mathematical properties of the function $F$}
\label{sec:reformulations}

We consider the same set-up as described at the beginning of Section~\ref{sec:main-result}.
The purpose of this section is to present a mathematical analysis of the properties described in Definition~\ref{def:properties-of-rhs} (in particular, parts~(\ref{def:scaling}) and~(\ref{def:monotonicity}) of the definition). 
To do so, we make a change of variables which is described in the following subsection and which will also play an essential role in the proof of the main theorem in Section~\ref{sec:proof-of-main-result}.

\subsection{Derivatives of $F$ vs.\ elasticities of $F$}

In the proof of the main theorem, we will employ the following coordinate transformation for equation~\eqref{eq:fixed-point-economic}.
It is thus important to reformulate the properties of $F$ listed in Definition~\ref{def:properties-of-rhs} 
in terms of the new variables; we will do so in Proposition~\ref{prop:reformulation-G} below.

Let $\clog: \bbR^N_{++} \to \bbR^N$ and $\cexp: \bbR^N \to \bbR^N_{++}$ denote the componentwise natural logarithm and exponential function, respectively, i.e., the mappings given by
\begin{align*}
	\big(\clog(x)\big)_j = \ln(x_j)
	\qquad \text{and} \qquad 
	\big(\cexp(z)\big)_j = \exp(z_j)
\end{align*}
for all $x \in \bbR^N_{++}$ and all $z \in \bbR^N$. 
Let us define 
\begin{align*}
	G := \clog \circ F \circ \cexp: \bbR^N \to \bbR^N.
\end{align*}
We will often use the change of variables $z := \clog(x)$. 
Then the fixed point equation~\eqref{eq:fixed-point-economic} in $\bbR^N_{++}$ is equivalent to the fixed point equation
\begin{align}
	\label{eq:fixed-point-log}
	z = G(z).
\end{align}

It is important for our purposes to understand the derivatives of the components of $G(z)$ in terms of the derivatives of the components of $F(x)$. For all indices $j,k \in \{1, \dots, N\}$ one has
\begin{align*}
    \frac{\partial G_j(z)}{\partial z_k}
    = 
    \frac{\partial \log F_j(x)}{\partial \log x_k}
    = 
    \frac{1}{F_j(x)} \frac{\partial F_j(x)}{\partial \log x_k}
    =
    \frac{x_k}{F_j(x)} \frac{\partial F_j(x)}{\partial x_k}.
\end{align*}

In other words, the partial derivative $\frac{\partial G_j(z)}{\partial z_k}$ is precisely the elasticity of $F_j(x)$ with respect to the variable $x_k$, and this elasticity is a strictly positive multiple of the partial derivative $\frac{\partial F_j(x)}{\partial x_k}$.
Hence, the entries of the Jacobi matrices
\begin{align*}
    DG(z) = \left(\frac{\partial G_j(z)}{\partial z_k}\right)_{j,k \in \{1, \dots, N\}}
    \quad \text{and} \quad 
    DF(x) = \left(\frac{\partial F_j(x)}{\partial x_k}\right)_{j,k \in \{1, \dots, N\}}
\end{align*}
are equal up to strictly positive multiples. 
This immediately implies assertions~(b)--(d) in the following proposition:

\begin{proposition}
    \label{prop:reformulation-G}
    Let $u \in \bbR^N$.
    \begin{enumerate}[\upshape (a)]
        \item 
        The following assertions are equivalent:
        \begin{enumerate}[\upshape (i)]
            \item 
            The function $F$ scales with exponent $u$, i.e., one has $F(c^u x) = c^u F(x)$ for all $x \in \bbR^N_{++}$ and all $c \in \bbR_{++}$.
            
            \item 
            One has $G(\lambda u + z) = \lambda u + G(z)$ for all $z \in \bbR^N$ and all $\lambda \in \bbR$.
            
            \item 
            One has $DG(z) u = u$ for all $z \in \bbR^N$.
        \end{enumerate}
        
        \item 
         The monotonicity behavior of the function $F$ is consistent with $u$ if and only if there exists a partition of $\{1,\dots,N\}$ into two disjoint subsets $\zeta_+$ and $\zeta_-$ such that $u_j \ge 0$ for all $j \in \zeta_+$ and $u_j \le 0$ for all $j \in \zeta_-$ and such that the following property holds for all $z \in \bbR^N$ and all indices $j,k \in \{1,\dots,N\}$:
        \begin{align*}
            \frac{\partial G_j(z)}{\partial z_k} & \ge 0 \quad \text{if both } j \text{ and } k \text{ are located in the same of the sets } \zeta_+,\zeta_-  \\ 
            \frac{\partial G_j(z)}{\partial z_k} & \le 0 \quad \text{if } j \text{ and } k \text{ are not located in the same of the sets } \zeta_+, \zeta_-.
        \end{align*}
        
        \item 
        The function $F$ connects all variables if and only if the modulus of the Jacobian matrix of $G$, $\modulus{DG(z)}$,
        is irreducible for each $z \in \bbR^N$.
        
        \item 
        The function $F$ exhibits self-interaction if and only if for each $z \in \bbR^N_{++}$ there exists an index $j \in \{1, \dots, N\}$ such that
        $\frac{\partial G_j(z)}{\partial z_j} \not= 0$.
    \end{enumerate}
\end{proposition}

\begin{proof}
    As already mentioned, assertions~(b)--(d) follow immediately from the observations made before the proposition, so let us prove~(a).
    Due to standard properties of the logarithm and the exponential function, (i) and~(ii) are equivalent, so we only have to prove that~(ii) and~(iii) are equivalent.
    
    ``(ii) $\Rightarrow$ (iii)''
    In the equality $G(\lambda u + z) = \lambda u + G(z)$, compute the derivative of both sides with respect to $\lambda$ and afterwards substitute $\lambda = 0$; this yields $DG(z)u = u$.
    
    ``(iii) $\Rightarrow$ (ii)'' 
    Fix $z \in \bbR^N$ and $\lambda \in \bbR$ and consider the path $\gamma: [0,1] \to \bbR^N$, $\gamma(t) = t \lambda u + z$ that parametrizes the straight line from $z$ to $\lambda u + z$. 
    Then it follows from the fundamental theorem of calculus (for functions with values in $\bbR^N$) that
    \begin{align*}
        G(\lambda u + z) - G(z) 
        & =
        G(\gamma(1)) - G(\gamma(0))
        = 
        \int_0^1 \frac{\dx}{\dx t} G(\gamma(t)) \dxInt t
        \\
        & = 
        \int_0^1 DG(\gamma(t)) \frac{\dx}{\dx t} \gamma(t) \dxInt t 
        = 
        \int_0^1 DG(\gamma(t)) \lambda u\dxInt t 
        = 
        \lambda \int_0^1 u \dxInt t
        = 
        \lambda u,
    \end{align*}
    where we used~(iii) for the penultimate equality.
\end{proof}

Note that if $u$ has at least one entry distinct from $0$ (in other words, $u$ is not the zero vector), then (a)(iii)
asserts that for each $z \in \bbR^N$, the number $1$ is an eigenvalue of the matrix $DG(z)$ with eigenvector $u$.

\subsection{The monotonicity condition}
\label{subsection:monotonicity-condition}

The purpose of this section is to discuss property~(\ref{def:monotonicity}) in Definition~\ref{def:properties-of-rhs} in more detail. 
We begin with a general spectral theoretic result about the domination of matrices. 

\begin{theorem}
    \label{thm:domination-of-matrices}
    Let $A, B \in \bbR^{N \times N}$ and $\modulus{A} \le B$, and assume that $B$ is irreducible. 
    Let $u \in \bbR^N$ be a non-zero vector such that $Au = u$.
    Then the following assertions are equivalent:
    \begin{enumerate}[\upshape (i)]
        \item 
        The matrix $B$ has spectral radius $1$.
        
        \item 
        One has $B\modulus{u} = \modulus{u}$.
        
        \item 
        There exists a non-zero vector $0 \le v \in \bbR^N$ such that $Bv = v$.
        
        \item 
        One has $\modulus{A} = B$ and there exists a partition of $\{1, \dots, N\}$ into two disjoints sets $\zeta_+,\zeta_-$ such that $u_j \ge 0$ for all $j \in \zeta_+$ and $u_j \le 0$ for all $j \in \zeta_-$ and such that the following property holds for all $j,k \in \{1,\dots,N\}$:
        \begin{align*}
            A_{jk} & \ge 0 \quad \text{if both } j \text{ and } k \text{ are located in the same of the sets } \zeta_+,\zeta_-  \\ 
            A_{jk} & \le 0 \quad \text{if } j \text{ and } k \text{ are not located in the same of the sets } \zeta_+, \zeta_-.
        \end{align*}
    \end{enumerate}
    If one (hence all) of these assertions holds, then all entries of $u$ are non-zero (and thus, the sets $\zeta_+$ and $\zeta_-$ in~(iv) are in fact uniquely determined by $u$), and the eigenspace of $A$ for the eigenvalue $1$ is one-dimensional (and thus spanned by $u$).
    Moreover, the matrices $A$ and $B$ are similar.%
    \footnote{I.e.,\ there exists an invertible $N \times N$-matrix $C$ such that $A = CBC^{-1}$.}
\end{theorem}

\begin{proof}
    We first note that if~(ii) holds, then every entry of $\modulus{u}$ is strictly positive since $B$ is irreducible
    (Theorem~\ref{thm:perron-frobenius-irreducible}(b)), and hence every entry of $u$ is non-zero.
    Now we prove the claimed equivalences.
    The geometric simplicity of the eigenvalue $1$ of $A$ will be shown at the end of the proof.
    
    ``(ii) $\Rightarrow$ (iii)'' 
    This follows immediately by choosing $v = \modulus{u}$.
    
    ``(iii) $\Rightarrow$ (i)''
    According to~(iii), the number $1$ is an eigenvalue of $B$ with a non-negative eigenvector. 
    So it follows from Theorem~\ref{thm:perron-frobenius-irreducible}(c) that $1 = \rho(B)$.
    
    ``(i) $\Rightarrow$ (ii)''
    Since $B$ is irreducible, so is its transposed matrix $B^\tp$.
    By applying Theorem~\ref{thm:perron-frobenius-irreducible}(b) to $B^\tp$, we see that there exists a vector $d \in \bbR^N_{++}$ such that $B^\tp d = d$, and hence $d^\tp B = d^\tp$.
    Now observe that
    \begin{align*}
        \modulus{u} = \modulus{Au} \le \modulus{A} \modulus{u} \le B \modulus{u},
    \end{align*}
    so the vector $B\modulus{u} - \modulus{u}$ is non-negative.
    As we have
    \begin{align*}
        d^\tp (B \modulus{u} - \modulus{u}) = d^\tp B \modulus{u} - d^\tp \modulus{u} = 0
    \end{align*}
    and as every entry of $d$ is positive, it follows that every entry of $B \modulus{u} - \modulus{u}$ is $0$. 
    Hence, $B\modulus{u} = \modulus{u}$.
    
    ``(ii) $\Rightarrow$ (iv)''
    We have already observed that, as~(ii) holds, every entry of $\modulus{u}$ is strictly positive and therefore, each entry of $u$ is non-zero.
    Moreover, by using the same computation as in the proof of the previous implication we see that 
    \begin{align*}
        \modulus{u} = \modulus{Au} \le \modulus{A} \modulus{u} \le B \modulus{u} = \modulus{u},
    \end{align*}
    so actually $\modulus{A}\modulus{u} = B \modulus{u}$, or equivalently $(\modulus{A}-B)\modulus{u} = 0$.
    As $\modulus{A} \le B$ and all entries of $\modulus{u}$ are strictly positive, this implies that $\modulus{A} = B$.
    
    Finally, define the sets
    \begin{align*}
        \zeta_+ := \big\{j \in \{1, \dots, N\}: \; u_j > 0 \big\}
        \quad \text{and} \quad 
        \zeta_- := \big\{j \in \{1, \dots, N\}: \; u_j < 0 \big\}.
    \end{align*}
    Since every entry of $u$ is non-zero, these two sets are a partition of $\{1, \dots, N\}$.
    Consider an index $j \in \{1, \dots, N\}$.
    It follows from $Au = u$ and $\modulus{u} = B\modulus{u} = \modulus{A}\modulus{u}$ that 
    \begin{align*}
        u_j = \sum_{k=1}^N A_{jk} u_k 
        \qquad \text{and} \qquad 
        \modulus{u_j} = \sum_{k=1}^N \modulus{A_{jk}} \modulus{u_k},
    \end{align*}
    so 
    \begin{align*}
        \modulus{\sum_{k=1}^N A_{jk} u_k} = \sum_{k=1}^N \modulus{A_{jk} u_k}.
    \end{align*}
    Thus, for a fixed $j$, either all the numbers $A_{jk}u_k$ are $\ge 0$, or all of them are $\le 0$ (since we have equality in the triangle inequality), and the same inequality is then true for the number $u_j$.
    
    Hence, if $j \in \zeta_+$, all the numbers $A_{jk}u_k$ are $\ge 0$ and thus we conclude that $A_{jk} \ge 0$ for $k \in \zeta_+$ and $A_{jk} \le 0$ for $k \in \zeta_-$.
    If, on the other hand, $j \in \zeta_-$, then all the numbers $A_{jk}u_k$ are $\le 0$, so $A_{jk} \le 0$ for all $k \in \zeta_+$ and $A_{jk} \ge 0$ for $k \in \zeta_-$.
    
    ``(iv) $\Rightarrow$ (ii)'' 
    As $\modulus{A} = B$, we only need to show that $\modulus{A}\modulus{u} = \modulus{u}$. 
    Fix $j \in \{1, \dots, N\}$. 
    We distinguish two cases:
    
    \emph{1st case: $j \in \zeta_+$.}
    In this case,
    \begin{align*}
        \modulus{u_j} = u_j = \sum_{k=1}^N A_{jk} u_k = \sum_{k \in \zeta_+} \underbrace{A_{jk}}_{\ge 0} \underbrace{u_k}_{\ge 0} + \sum_{k \in \zeta_-} \underbrace{A_{jk}}_{\le 0} \underbrace{u_k}_{\le 0} = \sum_{k=1}^N \modulus{A_{jk}} \modulus{u_k}.
    \end{align*}
    
    \emph{2nd case: $j \in \zeta_-$.}
    In this case,
    \begin{align*}
        \modulus{u_j} = -u_j = \sum_{k=1}^N -A_{jk} u_k = \sum_{k \in \zeta_+} -\underbrace{A_{jk}}_{\le 0} \underbrace{u_k}_{\ge 0} + \sum_{k \in \zeta_-} -\underbrace{A_{jk}}_{\ge 0} \underbrace{u_k}_{\le 0} = \sum_{k=1}^N \modulus{A_{jk}} \modulus{u_k}.
    \end{align*}
    So indeed $\modulus{u} = \modulus{A} \modulus{u}$.
    \medskip 
    
    Now assume that the equivalent assertions~(i)--(iv) are satisfied. 
    It remains to prove that $A$ and $B$ are similar (then the eigenspace of $A$ for the eigenvalue $1$ is automatically one-dimensional since the same is true for $\ker(1-B)$ as $0 \le B$ is irreducible and has spectral radius $1$, see Theorem \ref{thm:perron-frobenius-irreducible}(b)).
    However, since $B = \modulus{A}$ and since $1$ is an eigenvalue of $A$, the similarity of $A$ and $B$ follows from \cite[Proposition~1.6.4 on p.\,21]{Schaefer1974}.%
    \footnote{Note that there is a small inaccuracy in this reference: the quoted result is only true in the form stated there if both matrices have spectral radius $1$, which is the case in our setting. With different spectral radii the claimed equality exhibits the wrong scaling behavior under multiplication with positive scalars.}
\end{proof}

If $B$ is even primitive in the above theorem, then we get the following stronger property:

\begin{corollary}
    \label{cor:dom-primitive}
    In the situation of Theorem~\ref{thm:domination-of-matrices}, assume that the equivalent assertions~(i)--(iv) are satisfied and that $B$ is primitive. 
    Then $1$ is the only eigenvalue of $A$ with modulus $1$.
\end{corollary}

\begin{proof}
    This follows immediately from the similarity of $A$ and $B$
    and from the fact that $B$ is primitive, 
    since similar matrices always have the same set of eigenvalues.
\end{proof}

We can immediately reformulate Theorem~\ref{thm:domination-of-matrices} for our function $F$:

\begin{corollary} 
    \label{cor:signs-vs-spectral-radius}
    Let $u \in \bbR^N$ be a non-zero vector 
    and assume that $F$ scales with exponent $u$ and that $F$ connects all components.
    Then the following assertions are equivalent:
    \begin{enumerate}[\upshape (i)]
        \item 
        For each $z \in \bbR^N$ the modulus of the Jacobi matrix of $G$, $\modulus{DG(z)}$, has spectral radius $1$.
        
        \item 
        For each $z \in \bbR^N$ one has $\modulus{DG(z)} \modulus{u} = \modulus{u}$.
        
        \item 
        For each $z \in \bbR^N$ there exists a non-zero vector $0 \le v \in \bbR^N$ such that $\modulus{DG(z)} v = v$.
        
        \item 
        The monotonicity behavior of the function $F$ is consistent with $u$.
    \end{enumerate}
    If one (hence all) of these assertions holds, then all entries of $u$ are distinct from $0$ (and thus, the sets $\zeta_+$ and $\zeta_-$ in Definition~\ref{def:properties-of-rhs}(\ref{def:monotonicity}) are uniquely determined by $u$), and for each $z \in \bbR^N$ the eigenspace of $DG(z)$ for the eigenvalue $1$ is one-dimensional (and thus spanned by $u$).
\end{corollary}

\begin{proof}
    According to Proposition~\ref{prop:reformulation-G}(a), we have $DG(z)u = u$ for each $z \in \bbR^N$.
    Hence, the claim follows by applying, for each $z \in \bbR^N$, Theorem~\ref{thm:domination-of-matrices} to the matrices $A := DG(z)$ and $B := \modulus{DG(z)}$
    (note that $\zeta_+$ and $\zeta_-$ are uniquely determined by $u$ in Theorem~\ref{thm:domination-of-matrices} and will hence be the same no matter which $z$ we use in the definition of $A$ and $B$, allowing to derive assertion~(iv) from the others).
\end{proof}

\section{A stability result for non-expansive dynamical systems}
\label{section:stability-non-expansive}

The final component to our proof is the following stability result, loosely reminiscent of Banach's fixed point theorem \citep[which can, for instance, be found in][Theorem~3.48 on p.\,95]{AliprantisBorder2006}.

Recall that, for a metric space $(Z,\metric)$, a mapping $G: Z \to Z$ is called \emph{Lipschitz continuous with constant $L$} (where $L \ge 0$ is a real number) if the inequality
\begin{align*}
    \metric(G(z_1), G(z_2)) \le L \metric(z_1, z_2) 
\end{align*}
holds for all $z_1, z_2 \in Z$.
Banach's fixed point theorem assumes Lipschitz continuity with a constant $< 1$.
In the following, we require only Lipschitz continuity with the constant $1$ instead.
As a trade-off, we need to assume a priori existence of a fixed point which is locally asymptotically stable.
Moreover, the assumptions on the metric space are somewhat different than in Banach's fixed point theorem.

We use the following terminology from the theory of dynamical systems.
Let $(Z,\metric)$ be a metric space which is connected, let $G: Z \to Z$, and let $z^* \in Z$ be a fixed point of $G$, 
i.e., $G(z^*) = z^*$. 
The fixed point $z^*$ is called \emph{Lyapunov stable} if for every number $\varepsilon > 0$ 
there exists a number $\delta > 0$ such that, for all $z \in Z$ which satisfy $\metric(z, z^*) < \delta$,
ones has $\metric(G^n(z), G^n(z^*)) < \varepsilon$ for all integers $n \ge 0$.%
\footnote{For the special case of the function $F$ in equation~\eqref{eq:fixed-point-economic} we already recalled the definition of this property in the introduction.}
Moreover, the fixed point $z^*$ of $G$ is called \emph{locally attractive} if there exists a number $\delta > 0$
such that for all $z \in Z$ with $\metric(z,z^*) < \delta$ one has $G^n(z) \to z^*$ as $n \to \infty$.
Finally, $z^*$ is called \emph{locally asymptotically stable} if it is both Lyapunov stable and locally attractive.

Then for Lipschitz continuous functions with Lipschitz constant $1$, the following result expands the local asymptotic stability to a global attractivity property.

\begin{theorem}
    \label{thm:glob.al-stability-via-local-stability}
	Let $(Z,\metric)$ be a metric space which is connected,%
	\footnote{A metric space is called \emph{connected} if it cannot be written as the union of two non-empty open disjoint subsets.}
	let $G: Z \to Z$ be a Lipschitz continuous function with Lipschitz constant $1$, and let $z^* \in Z$ be a fixed point of $G$ which is locally asymptotically stable.
	
	Then $z^*$ is globally attractive, i.e., for each $z \in Z$ the sequence $G^n(z)$ converges to $z^*$ as $n \to \infty$. 
	In particular, $z^*$ is the only fixed point of $G$.
\end{theorem}
\begin{proof}
	Let $B$ denote the \emph{basin of attraction} of the fixed point $z^*$, 
	i.e., the set of all $z \in Z$ for which we have $G^n(z) \to z^*$ as $n \to \infty$ (note that $B$ is non-empty as $z^* \in B$). 
	We have to show that $B = Z$.
	
	Since $z^*$ is locally asymptotically stable, it easily follows that the set $B$ is open in $Z$.
	We now show that $B$ is also closed; since $Z$ is connected, this immediately implies $B = Z$ then \citep[see for instance][Theorem~(3.20)]{Armstrong1983}.
	So let $(z_k)_{k \in \bbN}$ be a sequence in $B$ which converges to a point $z \in Z$.
	In order to show that $G^n(z)$ converges to $z^*$ as $n \to \infty$, let $\varepsilon > 0$.
	
	There exists an index $k_0 \in \bbN$ such that $z_{k_0}$ is closer than $\varepsilon/2$ to $z$. 
	Since $z_{k_0}$ is located in the basin of attraction $B$, there exists $n_0 \in \bbN$ such that, for all $n \ge n_0$, the element $G^n(z_{k_0})$ is closer than $\varepsilon/2$ to $z^*$. 
	But this implies, also for all $n \ge n_0$,
	\begin{align*}
		\metric(G^n(z), z^*) 
		\le 
		\underbrace{\metric(G^n(z), G^n(z_{k_0}))}_{\le \metric(z, z_{k_0})} 
		+ 
		\metric(G^n(z_{k_0}), z^*)
		< 
		\frac{\varepsilon}{2} + \frac{\varepsilon}{2}
		= 
		\varepsilon.
	\end{align*}
	So we proved that, indeed, $G^n(z) \to z^*$ as $n \to \infty$, i.e., $z \in B$.
	Hence, $B$ is closed as claimed.
\end{proof}

\section{Proof of the mathematical main result}
\label{sec:proof-of-main-result}

For the proof of our main result, the following concept will be very useful. 
If $v \in \bbR^N_{++}$, then the \emph{gauge norm} with respect to $v$ is the norm $\norm{\argument}_v$ on the space $\bbR^N$ which is given by 
\begin{align*}
    \norm{z}_v 
    = 
    \min \{c \ge 0: \; \modulus{z} \le c v \}
    = 
    \max_{j = 1, \dots, N} \frac{\modulus{z_j}}{v_j}
\end{align*}
for all $z \in \bbR^N$.
Since all norms on $\bbR^N$ are equivalent,
\footnote{Two norms $\norm{\cdot}$ and $\normiii{\cdot }$ on some vector space $V$ are called \emph{equivalent} if there exist strictly positive constants $c,C > 0$ such that $c\norm{y} \le \normiii{y}  \le C \norm{y}$ for all $y \in V$.}
it suffices to prove all convergence results with respect to this norm for some vector $v \in \bbR^N_{++}$ of our choice. 
One key aspect of our proof is that we do not consider thereby a universal vector $v$, %
but that we choose a vector $v$ which is appropriate for the given function $F$. 
\bigskip 

From now on, let the assumptions of the main theorem be satisfied, i.e.\ let $u \in \bbR^N$ with at least one entry different from $0$ and assume that $F$ scales with exponent $u$, that the monotonicity behavior of $F$ is consistent with $u$, that $F$ connects all variables, and that $F$ exhibits self-interaction; also assume that the fixed point equation~\eqref{eq:fixed-point-economic} has a solution $x^* \in \bbR^N_{++}$.

As in Corollary~\ref{cor:signs-vs-spectral-radius},
we choose $v := \modulus{u}$.
It follows from Corollary~\ref{cor:signs-vs-spectral-radius} that all entries of $u$ are distinct from $0$, i.e., $v \in \bbR^N_{++}$.
As indicated above, we will work with the gauge norm $\norm{\argument}_v$ on $\bbR^N$.

We need the following lemma:

\begin{lemma}
    \label{lem:lipschitz-estimate}
    We have
    \begin{align*}
        \norm{G(z) - G(\tilde z)}_v \le \norm{z - \tilde z}_v
    \end{align*}
    for all $z, \tilde z \in \bbR^N$, i.e., the function $G$ is Lipschitz continuous with Lipschitz constant $1$ with respect to the Gauge norm $\norm{\argument}_v$ on $\bbR^N$.%
    \footnote{More precisely speaking, the function $G$ is Lipschitz continuous with constant $1$ with respect to the metric $\metric$ on $\bbR^N$ that is induced by the Gauge norm $\norm{\argument}_v$ -- i.e., the metric given by $\metric(z,\tilde z) := \norm{z-\tilde z}_v$ for all $z,\tilde z \in \bbR^N$.}
\end{lemma}
\begin{proof}
    According to Corollary~\ref{cor:signs-vs-spectral-radius}(ii), we have $\modulus{DG(z)}v = v$ for all $z \in \bbR^N$.
    Now, we argue similarly as in the proof of Proposition~\ref{prop:reformulation-G}(a):
    
    Fix $z, \tilde z \in \bbR^N$ and let $\gamma: [0,1] \to \bbR^N$ by the straight line which runs from $\tilde z$ to $z$, i.e., $\gamma(t) = \tilde z + t(z-\tilde z)$ for all $t \in [0,1]$.
    Then the fundamental theorem of calculus implies that
    \begin{align*}
        G(z) - G(\tilde z) 
        & =
        G(\gamma(1)) - G(\gamma(0))
        = 
        \int_0^1 \frac{\dx}{\dx t} G(\gamma(t)) \dxInt t
        \\
        & =
        \int_0^1 DG\big(\gamma(t)\big) \dot \gamma(t) \dxInt t
        = 
        \int_0^1 DG\big(\gamma(t)\big) (z - \tilde z) \dxInt t.
    \end{align*}
    Thus,
    \begin{align*}
        \modulus{G(z) - G(\tilde z)}
        & \le 
        \int_0^1 \modulus{DG\big(\gamma(t)\big)} \norm{z - \tilde z}_v \, v \dxInt t
        =
        \norm{z - \tilde z}_v \int_0^1 v \dxInt t
        =
        \norm{z - \tilde z}_v \, v
    \end{align*}
    (where we used $\modulus{z-\tilde z} \le \norm{z-\tilde z}_v v$ for the inequality at the beginning).
    This proves that
    \begin{align*}
        \norm{G(z) - G(\tilde z)}_v \le \norm{z - \tilde z}_v,
    \end{align*}
    as claimed.
\end{proof}

Now we can finally prove our main result.

\begin{proof}[Proof of the main theorem]
    We note that the vector $z^* := \clog(x^*) \in \bbR^N$ is a solution to the fixed point equation~\eqref{eq:fixed-point-log}.
    
    (a)
    To show the claimed up-to-scale uniqueness, let $\tilde z^* \in \bbR^N$ be another solution of~\eqref{eq:fixed-point-log} which is distinct from $z^*$.
    Let $\gamma: [0,1] \to \bbR^N$ be the straight line that runs from $\tilde z^*$ to $z^*$.
    Then, once again by the fundamental theorem of calculus\footnote{The usage of the fundamental theorem of calculus here is loosely reminiscent of the usage of the mean value theorem in the proof of \cite[Theorem~1]{AllenArkolakisLi2022}.
    We point out that those two important results from calculus are related in the sense that the mean value theorem for continuously differentiable functions can immediately be derived from the fundamental theorem of calculus.
    However, while the mean value theorem is only true for functions which map from an interval to $\bbR$, the fundamental theorem of calculus is also true for functions which map from an interval to $\bbR^N$, and this is a considerable advantage for our argument.}
    \begin{align*}
        z^* - \tilde z^* 
        =
        G(z^*) - G(\tilde z^*) 
        =
        \int_0^1 DG\big(\gamma(t)\big) \dxInt t \; (z^* - \tilde z^*),
    \end{align*}
    so $z^* - \tilde z^*$ is an eigenvector of the matrix $A := \int_0^1 DG\big(\gamma(t)\big) \dxInt t$ for the eigenvalue $1$.
    Moreover, we note that $u$ is clearly also an eigenvector of this matrix for the eigenvalue $1$.
    
    At the same time, $A$ is dominated by the matrix $B := \int_0^1 \modulus{DG\big(\gamma(t)\big)} \dxInt t$ in the sense that $\modulus{A} \leq B$.
    Since each of the matrices $\modulus{DG\big(\gamma(t)\big)}$ is irreducible (Proposition~\ref{prop:reformulation-G}(c)) and has the eigenvector $v$ (Corollary~\ref{cor:signs-vs-spectral-radius}), the matrix $B$ is also irreducible and has the eigenvector $v$.
    Thus it follows from Theorem~\ref{thm:domination-of-matrices} that the eigenspace $\ker(1-A)$ is one-dimensional, so $z^* - \tilde z^*$ is a multiple of $u$.
    This proves that all solutions are of the claimed form.
    
    On the other hand, all vectors of the form $z^* + \lambda u$ are indeed solutions of the fixed point equation~\eqref{eq:fixed-point-log} due to Proposition~\ref{prop:reformulation-G}(a).
    
    (b)
    Assume now that $F$ also exhibits self-interaction.
    We continue to use the notation introduced in the proof of~(a).
    We will prove~(b) by means of a quotient space argument.
    Let $U := \{\lambda u: \, \lambda \in \bbR\}$ denote the linear span of the vector $u$ (i.e., $U$ is a one-dimensional vector subspace of $\bbR^N$).
    We endow the quotient space $\bbR^N/U$ with the quotient norm induced by the norm $\norm{\argument}_v$ on $\bbR^N$.
    We denote the quotient norm by $\norm{\argument}_/$; it is defined as
    \begin{align*}
        \norm{[z]}_/ := \min \big\{ \|\tilde z\|_v: \, \tilde z \in [z] \big\}
    \end{align*}
    for each equivalence class $[z] \in \bbR^N/U$.
    
    Now we define a mapping $G: \bbR^N / U \to \bbR^N / U$ as follows: for each equivalence class $[z] \in \bbR^N/U$ we set
    \begin{align*}
        G_/([z]) = [G(z)].
    \end{align*}
    This mapping is well-defined (i.e., $G_/([z])$ does not depend on the choice of the representative $z$ of the equivalence class $[z]$) due to Proposition~\ref{prop:reformulation-G}(a)(ii). Indeed, if $z$ and $z^*$ belong to the same equivalence class, equivalently $z-z^* \in U$, then $z=\lambda u + z^*$ for some $\lambda \in \bbR$. This implies that
    \begin{align*}
        G(z) = G(\lambda u + z^*) = \lambda u + G(z^*),
    \end{align*}
    i.e., $G(z)-G(z^*)$ belongs to $U$.
    
    The point $[z^*] \in \bbR^N / U$ is obviously a fixed point of $G_/$. We will now show that, for every $[z] \in \bbR^N/U$, the iterates $G_/^n([z])$ converge to $[z^*]$ (in $\bbR^N/U$) as $n \to \infty$. By Proposition~\ref{prop:reformulation-G}(a)(iii), $u$ is a fixed vector of the matrix $DG(z)$ for each $z \in \bbR^N$ and hence, the linear mapping $DG(z)$ on $\bbR^N$ induces a linear mapping $DG(z)_/$ on the quotient space $\bbR^N / U$.
    A straightforward computation now shows that the mapping $G_/$ is differentiable and that, for every $z \in \bbR^N$, the derivative of $G_/$ at $[z]$ is equal to $DG(z)_/$.
    Moreover, the following holds for each $z \in \bbR^N$:
    According to Corollary~\ref{cor:signs-vs-spectral-radius} the matrix $\modulus{DG(z)}$ has spectral radius $1$ and the eigenspace of $DG(z)$ for the eigenvalue $1$ is equal to $U$.
    In addition, since the irreducible matrix $\modulus{DG(z)}$ has a non-zero diagonal entry (this is the only point where we use the assumption that $F$ exhibits self-interaction), it is primitive by Proposition~\ref{prop:primitivity}.
    Thus, it follows from Corollary~\ref{cor:dom-primitive} that $DG(z)$ has no eigenvalues on the unit circle except for $1$, and it follows from the similarity assertion in Theorem~\ref{thm:domination-of-matrices} that the eigenvalue $1$ of $DG(z)$ is algebraically simple.
    We can thus conclude that the linear mapping $DG(z)_/$ has no eigenvalues on the unit circle and hence has spectral radius $< 1$.%
    \footnote{More precisely, this argument works as follows: 
    by considering the Jordan normal form of $DG(z)$ and using that the Jordan block for the eigenvalue $1$ in this normal form has size $1 \times 1$ (as the eigenvalue $1$ of $DG(z)$ is algebraically simple), one sees that the spectrum of $DG(z)_/$ is precisely the spectrum of $DG(z)$ except for the number $1$, which is not an eigenvalue of $DG(z)_/$.
    Since $1$ is the only eigenvalue of $DG(z)$ which does not have modulus $<1$, we thus conclude that all eigenvalues of $DG(z)_/$ have modulus $<1$.}
    
    This implies, by the principle of linearized stability \cite[Theorem~3.3.52]{HinrichsenPritchard2005}, that the equilibrium $[z^*]$ of $G_/$ is locally asymptotically stable.
    Moreover, from the Lipschitz continuity of $G$ with Lipschitz constant $1$ one can easily derive that $G_/$ is also Lipschitz continuous with Lipschitz constant $1$. 
    Hence, it follows from Theorem~\ref{thm:glob.al-stability-via-local-stability} that indeed $G_/^n([z]) \rightarrow [z^*]$ as $n \rightarrow \infty$ for every $z \in \bbR^N$.
    
    Finally, we need to prove that this implies assertion~(b) for the map $G$ on $\bbR^N$.
    The Lyapunov stability of all fixed points of $G$ (and hence of $F$)  follows immediately from the Lipschitz continuity of $G$ (Lemma~\ref{lem:lipschitz-estimate}).
    To show the claimed convergence, fix $z \in \bbR^N$.
    Since $G_/^n([z])$ converges to $[z^*]$ as $n \to \infty$, there exists a sequence $(\lambda_n)_{n \in \mathbb{N}}$ in $\bbR$ such that $G^n(z) - \lambda_n u \to z^*$ as $n \to \infty$.
    
    We are now going to show that this implies that $(G^n(z))_{n \in \mathbb{N}}$ is a Cauchy sequence%
    \footnote{By definition of the notion \emph{Cauchy sequence} this means we have to show that for every $\varepsilon > 0$ there exists an index $n_0$ such that $\norm{G^{n_2}(z) - G^{n_1}(z)}_v < 2\varepsilon$ for all $n_1, n_2 \ge n_0$.}
    in $\bbR^N$, so let $\varepsilon > 0$.
    There exists an index $n_0$ such that $\norm{z^* + \lambda_{n_0} u - G^{n_0}(z)}_v < \varepsilon$.
	We now use that $z^* + \lambda_n u$ is a fixed point of $G$ for every $n \in \bbN$ and that $G$ is Lipschitz continuous with constant $1$.
	This implies that, for each $n \ge n_0$,
	\begin{align*}
		\norm{z^* + \lambda_{n_0} u - G^n(z)}_v
		& =
		\norm{G^{n-n_0}(z^* + \lambda_{n_0} u) - G^{n-n_0}(G^{n_0}(z))}_v
		\\
		& \le 
		\norm{z^* + \lambda_{n_0} u - G^{n_0}(z)}_v
		<
		\varepsilon.
	\end{align*}
	Thus, $\norm{G^{n_2}(z) - G^{n_1}(z)}_v < 2\varepsilon$ for all $n_1, n_2 \ge n_0$, which proves that $(G^n(z))_{n \in \mathbb{N}}$ is indeed a Cauchy sequence. 
	Since $\bbR^{N}$ is complete (with respect to any norm and thus, in particular, with respect to the norm $\norm{\argument}$), it follows that $(G^n(z))_{n \in \mathbb{N}}$ converges in $\bbR^{N}$. By the continuity of $G$, the limit is clearly a fixed point of $G$.
\end{proof}
\end{appendices}

\end{document}